\newif\ifcomments

\commentstrue

\documentclass[letterpaper,11pt,pdfa]{article}
\usepackage[in]{fullpage}

\usepackage{iftex}
\ifPDFTeX
  \usepackage[utf8]{inputenc}
  \usepackage[noTeX]{mmap}
  \makeatletter
  \@namedef{cmap@set@TS1}{}
  \@namedef{cmap@set@ts1}{}
  \makeatother
  \usepackage[T1]{fontenc}
\fi

\usepackage{amsmath}
\usepackage{amsfonts}
\usepackage{amsthm}
\usepackage{amssymb}
\ifPDFTeX\else
  \usepackage{unicode-math}
\fi
\usepackage{color}
\usepackage{microtype}
\usepackage{accsupp}
\AccSuppSetup{unicode}
\usepackage[bookmarks]{hyperref}
\usepackage[nameinlink]{cleveref}
\usepackage{aliascnt}

\newcommand{\email}[1]{\href{mailto:#1}{\texttt{#1}}}

\makeatletter
\newcommand{\TitleCase}[1]{\expandafter\TitleCase@i#1\@nil}
\def\TitleCase@i#1#2\@nil{\MakeUppercase{#1}#2}
\newcommand{\DeclarePluralForm}[2]{\expandafter\gdef\csname plural@#1\endcsname{#2}}
\newcommand{\PluralForm}[1]{%
  \ifcsname plural@#1\endcsname
    \csname plural@#1\endcsname
  \else
    #1s%
  \fi
}

\NewDocumentCommand{\DeclareTheoremWithAlias}{m o}{%
  \newaliascnt{#1}{theorem}%
  \newtheorem{#1}[#1]{\TitleCase{#1}}%
  \aliascntresetthe{#1}%
  \IfNoValueTF{#2}{%
    \edef\dtwa@plural{\PluralForm{#1}}%
  }{%
    \edef\dtwa@plural{#2}%
  }%
  \edef\dtwa@titleplural{\TitleCase{\dtwa@plural}}%
  \expandafter\crefname\expandafter{#1}{#1}{\dtwa@plural}%
  \expandafter\Crefname\expandafter{#1}{\TitleCase{#1}}{\dtwa@titleplural}%
}
\makeatother
\newtheorem{theorem}{Theorem}[section]
\DeclareTheoremWithAlias{definition}
\DeclareTheoremWithAlias{remark}
\DeclareTheoremWithAlias{lemma}
\DeclareTheoremWithAlias{corollary}[corollaries]
\DeclareTheoremWithAlias{proposition}
\DeclareTheoremWithAlias{claim}
\DeclareTheoremWithAlias{observation}
\DeclareTheoremWithAlias{fact}

\usepackage[backend=biber,style=alphabetic,minalphanames=3,maxalphanames=4,maxnames=99,backref=true]{biblatex}

\DeclareFieldFormat{eprint:iacr}{Cryptology ePrint Archive: \href{https://ia.cr/#1}{\texttt{#1}}}
\DeclareFieldFormat{eprint:iacrarchive}{Cryptology ePrint Archive: \href{https://eprint.iacr.org/archive/#1}{\texttt{#1}}}
\AtEveryBibitem{%
 \clearlist{address}
 \clearfield{date}
 \clearfield{isbn}
 \clearfield{issn}
 \clearlist{location}
 \clearfield{month}
 \clearfield{series}

 \ifentrytype{book}{}{%
  \clearlist{publisher}
  \clearname{editor}
 }
}

\usepackage{enumitem}
\setlist[description]{noitemsep}
\setlist[enumerate]{noitemsep}
\setlist[itemize]{noitemsep}

\usepackage{soul}
\usepackage{xcolor}
\usepackage{xparse}
\makeatletter
\ExplSyntaxOn
\cs_new:Npn \white_text:n #1
{
  \fp_set:Nn \l_tmpa_fp {#1 * .01}
  \llap{\textcolor{white}{\the\SOUL@syllable}\hspace{\fp_to_decimal:N \l_tmpa_fp em}}
  \llap{\textcolor{white}{\the\SOUL@syllable}\hspace{-\fp_to_decimal:N \l_tmpa_fp em}}
}
\NewDocumentCommand{\whiten}{ m }
{
  \int_step_function:nnnN {1}{1}{#1} \white_text:n
}
\ExplSyntaxOff

\NewDocumentCommand{ \varul }{ D<>{5} O{0.2ex} O{0.1ex} +m } {%
  \begingroup
  \setul{#2}{#3}%
  \def\SOUL@uleverysyllable{%
    \setbox0=\hbox{\the\SOUL@syllable}%
    \ifdim\dp0>\z@
    \SOUL@ulunderline{\phantom{\the\SOUL@syllable}}%
    \whiten{#1}%
    \llap{%
      \the\SOUL@syllable
      \SOUL@setkern\SOUL@charkern
    }%
    \else
    \SOUL@ulunderline{%
      \the\SOUL@syllable
      \SOUL@setkern\SOUL@charkern
    }%
    \fi}%
  \ul{#4}%
  \endgroup
}
\makeatother

\newcommand{\E}{\mathop{\mathbb{E}}}

\usepackage{braket}
\usepackage{mleftright}

\newcommand{\mparen}[1]{\mleft(#1\mright)}

\newcommand{\norm}[1]{\left\lVert#1\right\rVert}

\newcommand{\Tr}{\mathrm{Tr}}
\newcommand{\ketbra}[2]{\ket{#1}\!\bra{#2}}
\newcommand{\proj}[1]{\ketbra{#1}{#1}}
\newcommand{\Init}{\mathsf{Init}}

\ifcomments
  \newcommand{\luowen}[1]{{\color{blue}Luowen: #1}}
\else
  \newcommand{\luowen}[1]{}
\fi

\begin{document}

\title{Impersonating Quantum Secrets over Classical Channels}
\author{Luowen Qian\thanks{NTT Research, Inc. \email{luowen.qian@ntt-research.com}} \and Mark Zhandry\thanks{Stanford University \& NTT Research, Inc. \email{mzhandry@stanford.edu}}}
\date{}
\maketitle

\begin{abstract}
We show that a simple eavesdropper listening in on classical communication between potentially entangled quantum parties will eventually be able to impersonate any of the parties.
Furthermore, the attack is efficient if one-way puzzles do not exist.
As a direct consequence, one-way puzzles are implied by reusable authentication schemes over classical channels with quantum pre-shared secrets that are potentially evolving.

As an additional application, we show that any quantum money scheme that can be verified through only classical queries to \emph{any} oracle cannot be information-theoretically secure.
This significantly generalizes the prior work by Ananth, Hu, and Yuen (ASIACRYPT'23) where they showed the same but only for the specific case of random oracles.
Therefore, verifying black-box constructions of quantum money inherently requires coherently evaluating the underlying cryptographic tools, which may be difficult for near-term quantum devices.
\end{abstract}

\section{Introduction}

In this work, we consider the following setting: Alice and Bob start off in possession of polynomial-sized secret, potentially-entangled quantum states. They then engage in a perpetual interactive protocol over a classical public channel. An eavesdropper Eve listens in on all of their communication. Eve knows the protocol, but does not (initially) know the shared quantum state of Alice and Bob. Eventually, Eve will try to impersonate (say) Alice to Bob. In fact, Eve will actually attempt to produce a quantum state which can be swapped for Alice's state. Alice and Bob then interact for a few more rounds, just with Alice's state replaced with Eve's state; Eve's goal is to get Bob to fail to notice that Alice's state was swapped, which we call impersonating Alice.
Bob's goal is to detect the swap.

This impersonation setting is natural, and as we will see encompasses a range of interesting applications. Our main result is that such a protocol cannot be information-theoretically secure. That is, given enough rounds of interaction in the protocol, a computationally unbounded Eve will eventually be able to impersonate Alice with all but inverse-polynomial probability. We can even say something stronger, namely, that any protocol of this form secure against \emph{efficient} quantum adversaries inherently requires \emph{cryptographic} assumptions. Slightly more formally:

\begin{theorem}[Informal]\label{thm:maininf} There exists a (potentially inefficient) Eve which will be able to impersonate Alice after a sufficiently large polynomial number of interactions. Moreover, if one-way puzzles do \emph{not} exist, then Eve runs in quantum polynomial time.
\end{theorem}
One-way puzzles~\cite{STOC:KhuTom24} are a foundational quantum cryptographic primitive. They are considered fairly mild: for example, one-way puzzles are implied by one-way functions. On the other hand, they are believed not to be the mildest assumption, as concepts like quantum commitments are believed to be even milder. One-way puzzles are efficiently broken with query access to $\mathsf{PP} \subseteq \mathsf{PSPACE}$.

We now explain two seemingly unrelated applications of \Cref{thm:maininf}. We first consider the task of repeated authentication over classical channels, where Alice and Bob's quantum key can be an arbitrary polynomial-sized entangled quantum state. The number of authentications is unbounded. An immediate consequence of Theorem~\ref{thm:maininf} is:
\begin{corollary}[informal]\label{cor:authinf} If one-way puzzles do not exist, then neither does quantum authentication over classical channels. In particular, quantum authentication over classical channels cannot be information-theoretically secure.
\end{corollary}

For our second application, we turn to public-key quantum money. Using Theorem~\ref{thm:maininf}, we show that the verifier in public-key quantum money must inherently make quantum queries to the underlying cryptographic building blocks:
\begin{corollary}[informal]\label{cor:qminf} There is no black box construction of public key quantum money from \emph{any} cryptographic tool (or combination of tools), where the verifier only makes classical queries to the cryptographic tool(s).
\end{corollary}
This significantly expands a similar result of \cite{AC:AnaHuYue23} who show the same, but only for very mild cryptographic tools from symmetric-key cryptography. Our result completely resolves the question of classical-query verification of public-key quantum money. 

We now discuss some motivation for studying these notions of quantum authentication and quantum money.

\subsection{Motivation}

\paragraph{Authentication and QKD.} Quantum key distribution (QKD)~\cite{BB84} allows for Alice and Bob to establish a shared secret key over a public channel. As such, it serves a similar role as classical public key encryption (PKE). Remarkably, however, despite classical PKE seemingly requiring very strong algebraic assumptions, QKD in its usual formulation requires no assumptions at all. QKD thus sparked significant interest in information-theoretically secure quantum cryptographic protocols.

That said, typically even quantum cryptography requires computational hardness, though hopefully weaker than classical hardness. For example, just like classical key exchange, QKD still assumes a classical authenticated channel.
Establishing such a channel is typically via public key infrastructure, or in the absence of (or if we are unwilling to trust) a setup, a short classical secret shared between the two parties which can then be used to authenticate the classical messages exchanged.
One way of doing this is via computationally secure MACs at the cost of only achieving everlasting security; the other way is via key recycling and information-theoretic MACs \cite{Mueller-QuadeRenner09,PortmannRenner14}, but this approach uses quantum communication and risks the adversary carrying out a denial-of-service (DoS) attack to exhaust the secret randomness.
Thus it becomes natural to investigate if there is a better information-theoretic approach to authentication, potentially by leveraging quantum resources.

Our version of quantum authentication can be seen as a natural generalization of classical authentication to the weakest but still meaningful setting, where (1) communication can be interactive, (2) Alice and Bob now have arbitrary local quantum computation at their disposal, and (3) they can leverage entangled quantum keys that update over time.
Without loss of generality, we also simplify the setting so that the same message is authenticated each time, which is non-trivial due to stateful interaction.

In this version, even with only classical communication, it is a priori unclear if this is possible information-theoretically due to the many relaxed requirements on the protocol itself compared to classical MACs.
The use of quantum keys and computation means that such authentication no longer necessarily implies one-way function, and can conceivably be based on far milder assumptions.
Our \Cref{cor:authinf} shows that any such quantum authentication still requires either quantum communication channels (such as key recycling) or computational hardness at least as strong as one-way puzzles, even if we are only concerned with attacks that are passive for all but one round and succeed with high probability.

\paragraph{Quantum Computation + Classical Communication Protocols.}
There has been a significant interest recently in hybrid classical/quantum protocols where the communication is classical, but some or all of the parties have local quantum computation. A famous example is the  cryptographic proof of quantumness of~\cite{FOCS:BCMVV18}. Another example is adaptations of Merkle puzzles to the quantum setting~\cite{BraSal08,BHKKLS19}. More recently, there have been efforts to lower-bound the hardness needed for such protocols~\cite{C:ACCFLM22,C:LLLL24,TCC:AnaGulLin25,ITCS:AnaKalYue25}.
The most related work is \cite{EC:QiaRaiZha25} where it is shown that QKD (or key exchange) requires either one-way puzzles or quantum communications; our first application extends this to the authentication setting as well.

\paragraph{Barriers for Quantum Money.} Quantum money is the most fundamental primitive in the field of uncloneable cryptography, where the uncloneability of quantum mechanics is leveraged for cryptographic purposes. In order to be useful, it has long been recognized that quantum money should admit public verification, where anyone can verify banknotes, while only the mint can create them. A number of constructions of such ``public key'' quantum money exist.

A fundamental question is to understand the types of cryptographic assumptions needed to realize public key quantum money. Unfortunately, public-key quantum money is only known to be constructed from strong and/or untested computational assumptions. 

A recent work by Ananth, Hu, and Yuen~\cite{AC:AnaHuYue23} initiates the study of lower-bounds on quantum money, attempting to justify this state-of-affairs. Very roughly, they show the following:
\begin{theorem}[{\cite[Theorem 1]{AC:AnaHuYue23}}, informal]\label{thm:AHYinf}
	There is no ``black box'' construction of public key quantum money from collision-resistant hash functions, where the verifier only makes classical queries to the hash function.
\end{theorem}
Note that while hash functions are considered relatively mild --- and certainly far milder than the assumptions currently used to construct public key quantum money --- \Cref{thm:AHYinf} is a notable improvement on what was previously known in terms of lower-bounds, which was absolutely nothing.

On the other hand, \Cref{thm:AHYinf} actually does not rule out arbitrary constructions from collision-resistance, but instead only those where the verifier makes classical queries. They left it as an explicit open problem to remove this limitation, which remains open.

Rather than resolving their open question,
our \Cref{cor:qminf} actually significantly extends ~\cite{AC:AnaHuYue23} in a different direction to handle black-box constructions from \emph{truly arbitrary} cryptographic building blocks, while still keeping the restriction to classical-query verifiers. 

One interpretation of our results is that \cite{AC:AnaHuYue23} is actually saying little about building quantum money from collision resistance. Instead, it is initiating the study of classical-query verifiers in quantum money. In this view, the classical-query verifier is the main feature of their impossibility, and the main limitation is that it only handles quantum money built from hash functions, instead of stronger building blocks. Our result fully generalizes their result, entirely refuting the possibility of classical query verifiers.

We believe classical-query verifiers in quantum money are well-motivated, which we now elaborate on. A natural question, distinct from the computational hardness question above, is what quantum resources are required for a public key quantum money scheme. In particular, we imagine the bottleneck for widespread deployment comes from the verifier: while the mint is centralized and may in the not-too-distant future have the resources to build a full fault-tolerant quantum computer in order to create banknotes, once the banknotes are distributed they will be verified by ordinary users on comparatively lightweight devices such as cell phones. Can public key quantum money exist where the verification is light on quantum resources; for example, can NISQ-era devices verify public key quantum money?
In comparison, we already have experimental demonstrations of QKD.

Observe that all current quantum computers, despite being capable of running relatively large processes, are completely incapable of running any classical cryptographic primitive in superposition: in fact, the main bottleneck for many quantum factoring algorithms lies in coherently multiplying two large numbers in finite fields.
Thus, a very natural metric to separate light-weight quantum devices from heavy-weight quantum devices is whether or not they are able to coherently run classical cryptographic primitives. 

In light of this delineation, one way to meaningfully capture a NISQ-verifier is to make the verifier have only classical queries to the underlying cryptography. No existing scheme has such a property.

Under this motivation and interpretation, ~\cite{AC:AnaHuYue23} give an initial barrier to near-term public key quantum money. Then our work significantly expands~\cite{AC:AnaHuYue23} from collision-resistance to \emph{anything}, definitively resolving the question of quantum verification queries. Thus, any implementation of public key quantum money inherently requries verifiers with quantum computers capable of coherently running the cryptographic tools. This in particular hints that quantum money verifiers likely require full fault-tolerant quantum computers.

Note that our \Cref{cor:qminf} even considers tools that are inherently non-black box, such as zero knowledge proofs or indistinguishability obfuscation (iO), which may be run on other cryptographic building blocks. As such, our result, while being nominally ``black box'' captures almost all cryptographic techniques, including even a wide range of non-black-box techniques. For a more detailed discussion of this point, see \Cref{sec:interpretation}.

\begin{remark} We note that the dual question of making the mint NISQ or even classical has been been studied in a few works. A number of protocols such as \cite{STOC:AarChr12} and follow-up works have quantum mints, but the mint only makes classical queries to the underlying cryptographic tools. Even more, some works~\cite{RadSat19,STOC:Shmueli22} show that it actually is possible to have a fully classical mint, under suitable cryptographic assumptions. On the lower-bound side,~\cite{EC:Zhandry25a} shows a black-box separation between public key quantum money and various cryptographic tools, where the mint is restricted to making classical queries. Thus, classical-query mints are more prevalent in the literature. However, from a resource-limitation perspective, it makes much more sense to try to minimize the quantum resources required for the verifier.
\end{remark}

\paragraph{Cryptographic advantage of quantum queries.}
Interestingly, this result also seems to give us an example where quantum queries to an oracle can be \emph{cryptographically} more powerful than classical queries.
Recall that there exists a classical oracle relative to which information-theoretically secure public-key quantum money exists \cite{STOC:AarChr12}.
However, our result shows that any such scheme must make use of quantum queries to the oracle in the verification algorithm.
This reveals that, just like query complexity separations \cite{FOCS:Simon94,BV97}, quantum superposition access to a classical function could also yield better security in the cryptographic setting.

\section{Technical Overview}

\subsection{A classical intuition} 

Consider the classical version of our setup, where Alice and Bob share correlated classical strings. Let us call Alice's string $k$, which has length $n$. We will also assume without loss of generality that $k$ persists throughout the interaction, though Alice may also keep separate local dynamic storage. Alice and Bob now engage in a public dialog over a channel seen by the eavesdropper Eve.
Eve knows the protocol, but does not know $k$.

We want to show that Eve will eventually be able to impersonate Alice to Bob. In particular, Eve will eventually be able to devise a key $k'$ which, if swapped for Alice's key $k$, will be undetectable by Bob for several rounds of interaction.

The intuition is that, for each message from Alice to Bob (or vice versa), one of the following happens:
\begin{itemize}
    \item Conditioned on the previous messages, this message has significant correlation with $k$. In this case, the message from Alice reveals some information about $k$.
    \item Conditioned on the previous message, this message is almost independent of $k$.
\end{itemize}
The key point is that the first case can only happen a bounded number of times, since $k$ contains only $n$ bits of information. That is, Alice will eventually run out of entropy, at which point her next message is just determined by the previous messages and is independent of her key. But in this case, Eve can sample $k'$ from the distribution of Alice's key conditioned on the messages seen so far, and $k'$ will result in (approximately) the same output distribution as $k$. Thus, Bob cannot tell if Alice has $k$ or $k'$.

Now, precisely defining what it means for Eve to impersonate Alice requires some care. First, $k$ may contain a number $t$ that specifies Alice to behave one way for the first $t-1$ messages, and then behave a different way for message $t$ and beyond. If Eve tries to impersonate at message $t$, Eve's view will be independent of whether $k$ contains $t$ or $t+1$. But these two cases result in very different behavior by Alice, meaning Eve cannot impersonate Alice correctly.

On the other hand, Alice's key can only contain at most roughly $n$ of these trigger points. By randomly choosing the point at which to impersonate from a set of size much larger than $n$, Eve will evade the triggers, except with inverse-polynomially small probability.

More generally, Eve will only be able to impersonate Alice for a short amount of time (since if she goes for too long, she may hit another trigger). But Eve will be able to make that amount of time arbitrarily long, by waiting to see more messages. Moreover, Eve's impersonation will be unable to achieve a negligible error relative to Alice's, but the error can be made arbitrarily inverse-polynomially small by eavesdropping for longer. Such a statement can be proved by carefully employing classical information-theory inequalities.

A bit more formally, let $S$ be the random variable for Alice's secret $k$, and $A_1, A_2, \ldots$ be the messages sent.
The correlation in the two cases in our intuition can be captured by $I(S; A_i \mid A_{<i})$, the mutual information between $S$ and the latest message $A_i$ conditioned on the prior transcript $A_{<i}$.
When the conditional mutual information is low, we can argue the success of forgery by invoking Pinsker's inequality.

\subsection{Moving to quantum keys}

When we move to quantum keys, the high-level intuition remains the same but the situation gets somewhat trickier. Due to the observer effect in quantum mechanics, each message generated by Alice and Bob potentially changes their state. In fact, their local states can end up with very high amounts of entropy which grow with the number of messages.
Because Alice's state keeps expanding, it becomes difficult to argue that eventually her state has low conditional mutual information with the new message.
In fact, the new message could always have a maximum number of bits of conditional mutual information with her updated state if, e.g., Alice just alternates standard- and Hadamard-basis measurements on her state and reports the outcome as the message.

To illustrate our ideas, we start with a simplified information-theoretic task that captures the heart of our impersonation attack. Consider a hidden register $X$ (together with a purifying environment $E$) that is processed over $n$ rounds. In each round, the honest party applies some isometry to $X$ and then performs a coherent measurement that yields a classical outcome $Q_i$ released to the outside world, while keeping any private workspace. Our goal in this toy model is to show that in a sufficiently long interaction that exposes only the classical transcript $Q_1,\dots,Q_n$, there must exist a round $i$ whose fresh leakage $Q_i$ reveals only a vanishing amount of \emph{new} information about $E$ beyond what is already contained in the past transcript $Q_{<i}$. Equivalently, $Q_i$ can be approximately simulated from $Q_{<i}$ alone. This simplified version is good enough to show that Eve can at least simulate a single one of Alice's messages.

\begin{observation}
    The entropy of Alice's state conditioned on the transcript $H(X \mid Q_{\le i})$ can never increase, that is, measuring an unknown quantum state can never increase our uncertainty about the state.
\end{observation}
\begin{proof}
Let $\rho_0$ be the initial state over system $X$.
In each iteration $i = 1, \ldots, n$, we start with $\rho_{i - 1}$ and apply the following operations to get $\rho_i$:
\begin{enumerate}
    \item Apply an isometry $M_{X \to X, Q_i}$. (For notational convenience, we reuse the symbol $X$ for the registers across rounds; the isometry $M$ may increase the dimension of the private register.)
    \item Apply a coherent measurement isometry $CNOT_{Q_i \to Q_i, Q_i'}$.
\end{enumerate}

Then for all $i$,
\begin{align*}
    H(X \mid Q_{\le i - 1})_{\rho_{i - 1}}
      &= H(X Q_i Q_i' \mid Q_{<i})_{\rho_i} \\
      &\ge H(XQ_i \mid Q_{<i})_{\rho_i} - H(Q_i' \mid Q_{<i})_{\rho_i} \\
      &= H(XQ_i \mid Q_{<i})_{\rho_i} - H(Q_i \mid Q_{<i})_{\rho_i} \\
      &= H(X \mid Q_{\le i})_{\rho_i},
\end{align*}
where the inequality is by the first part of the Araki--Lieb inequality (\Cref{lem:araki-lieb}).
\end{proof}

This gives a chain of inequalities $H(X)_{\rho_0} \ge H(X \mid Q_1)_{\rho_1} \ge \cdots \ge H(X \mid Q_{\le n})_{\rho_n} \ge 0$.
Taking the RHS of the first and the fourth line, we get that for every $n$, we have that there exists $i \le n$ such that
\begin{equation}
    \label{eq:rhoi-decrement-small}
    H(XQ_iQ_i' \mid Q_{<i})_{\rho_i} - H(X \mid Q_{\le i})_{\rho_i} \le H(X)_{\rho_0}/n.
\end{equation}

Now we take the purifying register $E$ for $X$ into account.
Specifically, at the beginning we have $H(XE)_{\rho_0} = 0$ by definition.
For transitioning from $\rho_{i - 1}$ to $\rho_i$, we are only tracing out $Q_i'$, thus $H(XE \mid Q_{\le i})_{\rho_i} = 0$ holds for all $i$ as well by induction.
Furthermore, $H(XQ_i Q_i' E \mid Q_{<i})_{\rho_i} = H(XE \mid Q_{<i})_{\rho_{i - 1}} = 0$ as well.
Therefore, by \eqref{eq:rhoi-decrement-small},
\[
I(E;Q_i\mid Q_{<i})_{\rho_i} = H(E \mid Q_{<i})_{\rho_i} - H(E \mid Q_{\le i})_{\rho_i} = H(XQ_iQ_i' \mid Q_{<i})_{\rho_i} - H(X \mid Q_{\le i})_{\rho_i} \le H(X)_{\rho_0}/n.
\]
Intuitively, this inequality says that eventually the new measurement results are almost independent of the environment register.
Thus, it is possible to reproduce $Q_i$ from $Q_{< i}$ without the knowledge of $E$.
This intuition is formalized through quantum Pinsker's inequality.

\begin{lemma}[{Quantum Pinsker's inequality \cite[Theorem 3.3]{HOT81-pinsker}}]
    Let $\rho_{XY}$ be a density matrix and $\sigma_{XY} := \rho_X \otimes \rho_Y$.
    Then $\norm{\rho - \sigma}_1 \le \sqrt{\frac2{\ln 2} \cdot I(X;Y)_\rho}$.
\end{lemma}

Collecting the above, we get that $\frac12 \norm{(\rho_i)_{EQ_i} - (\rho_i)_{E} \otimes (\rho_i)_{Q_i}}_1 \le \sqrt{\frac{H(X)_{\rho_0}}{2
\ln 2 \cdot n}}$.
This implies that $Q_i$ is simulatable from $Q_{<i}$ up to a $\sqrt{\frac{H(X)_{\rho_0}}{2\ln 2 \cdot n}}$ loss in trace distance.

\paragraph{Extending to our general impersonation task.} Extending to the impersonation task outlined at the beginning requires a bit more work. First, we need to handle impersonation for several rounds of messages. Then, we need to simulate Alice's actual state, not just her messages. To simulate her state, we sample from the conditional distribution of her state conditioned on the messages seen so far. We need that this faked state still generates the approximate distribution of messages Alice would produce.
One idea is that we look at a collection of messages and identify when the total conditional mutual information is low.
However, this is still subtle because now Alice could generate multiple messages after receiving the swapped state, whereas above we only argued that the first message would be fine, or alternatively if we update our belief and forge a new state every time we see a new message.

A more challenging issue is that we also need to take Bob's part of the state as well as Bob's messages into account.
It is possible that Bob's message may reveal some information to Alice about her state so that the distinguishing the swapped state becomes easier.
In the full proof, we develop a careful hybrid argument to address the full technicalities.

\subsection{Applications}

The impossibility of information-theoretic quantum authentication over classical channels in \Cref{cor:authinf} is an almost immediate application of our impersonation attack from \Cref{thm:maininf}.

To apply \Cref{thm:maininf} to obtain our quantum money impossibility in \Cref{cor:qminf}, we cast the quantum money verifier making classical queries to the underlying cryptographic tool as an instance of our impersonation game, with Alice being the public verifier, Alice's state being the quantum banknote, and Bob being the bank and the cryptographic tool. The eavesdropper is then just watching the queries made by the verifier to the oracle, and is able to create a new quantum banknote state (just from seeing the queries) which passes verification. Since Alice still has the quantum banknote state, the combined Alice-plus-Eavesdropper has turned one banknote into two, breaking the security of the quantum money scheme.

\section{Preliminaries}

\subsection{Notation and Conventions}

\paragraph{Security and asymptotics.} Let $\lambda \in \mathbb{N}$ denote the security parameter. An algorithm is quantum polynomial time (QPT) if its running time is polynomial in $\lambda$. A function $\nu(\lambda)$ is negligible if for every polynomial $p$ there exists $\lambda_0$ such that for all $\lambda \ge \lambda_0$, $\nu(\lambda) < 1/p(\lambda)$. Throughout, we will typically suppress the security parameter, and leave it as an implicit input to all functions and algorithms.

\paragraph{Information measures.} All logarithms are base 2. For a density operator $\rho_X$, $H(X)_\rho := -\Tr(\rho_X \log \rho_X)$. For a joint state $\rho_{XY}$, the (quantum) conditional entropy is $H(X\mid Y)_\rho := H(XY)_\rho - H(Y)_\rho$. Mutual information is $I(X;Y)_\rho := H(X)_\rho + H(Y)_\rho - H(XY)_\rho$ and conditional mutual information is $I(X;Y\mid Z)_\rho := H(XZ)_\rho + H(YZ)_\rho - H(Z)_\rho - H(XYZ)_\rho$. The trace norm is $\norm{A}_1 := \Tr\sqrt{A^\dagger A}$ and the trace distance between states $\rho,\sigma$ is $\tfrac12\norm{\rho - \sigma}_1$.
For a multipartite state $\rho_{XYZ}$, we denote reduced states by subscripts, e.g., $(\rho)_{XY} := \Tr_Z(\rho)$ and $\rho_X := \Tr_{YZ}(\rho)$. Thus, $\rho_X$ may implicitly involve tracing out all registers other than $X$.

The data processing inequality for trace distance states that for any quantum channel $\Psi$ and states $\rho, \sigma$, $\norm{\rho - \sigma}_1 \ge \norm{\Psi(\rho) - \Psi(\sigma)}_1$.

\paragraph{Classical conditioning.} If $Q$ is classical, we write $\rho\mid Q{=}q$ (or $\rho\mid q$ if $Q$ is clear from the context) for the (normalized) post-measurement state and abbreviate $Q_{\le i} := (Q_1,\ldots,Q_i)$ and $Q_{< i} := (Q_1,\ldots,Q_{i-1})$. For cq states, conditioning on $Q$ is equivalent to averaging over outcomes as above: $H(X\mid Q)_\rho = \sum_q p_q H(X)_{\rho\mid q}$ and similarly for conditional mutual information.

For a classical-quantum (cq) state of the form $\rho_{XQ} = \sum_q p_q\, \rho_X^{(q)} \otimes \ket{q}\!\bra{q}$, one has
\begin{align*}
  H(X\mid Q)_\rho &= \sum_q p_q\, H\big(\rho_X^{(q)}\big),\\
  I(X;Y\mid Q)_\rho &= \sum_q p_q\, I\big(X;Y\big)_{\rho\mid Q{=}q};
\end{align*}
see, e.g., \cite[Sec.~11]{Wilde2017-QIT}.

\begin{lemma}[{Araki--Lieb inequality \cite[(3.1)]{AL70-entropy}}]
\label{lem:araki-lieb}
	Let $\rho_{ABC}$ be any mixed state, then
	$|H(A)_\rho - H(B)_\rho| \le H(A B)_\rho \le H(A)_\rho + H(B)_\rho$.
\end{lemma}

\paragraph{Deferred measurement and coherent copies.} Measuring a computational-basis observable and recording the outcome in a classical register can be equivalently expressed by an isometry $\ket{x} \mapsto \ket{x}\ket{x}$ (e.g., a CNOT) given that the second register is not used in the future.
By the principle of deferred measurement, any measurement that only classically controls later computation can be assumed to occur at the end without loss of generality.

\paragraph{Communication channels and oracle access.} A public authenticated classical channel provides integrity but not confidentiality: the adversary observes all messages but cannot modify them. We distinguish classical oracle access (queries are classical strings) from quantum (superposition) access, which is the unitary described by $\ket x\ket y \mapsto \ket x\ket{y + f(x)}$.

\subsection{Quantum Cryptography}

We now consider the strong classical$\rightarrow$quantum extrapolation task, which is a strengthening of the classical$\rightarrow$quantum task defined in \cite{EC:QiaRaiZha25}.

\newcommand{\Gen}{\mathsf{Gen}}
\newcommand{\RegA}{A}
\newcommand{\RegB}{B}
\newcommand{\Adv}{\mathsf{Adv}}
\begin{definition}
  A classical$\rightarrow$quantum extrapolation problem is specified by a circuit $\Gen$ that produces a pure state, which can be written as
  \[
      \ket{\Gen} = \sum_{s} \alpha_s \ket{s}_{\RegA} \otimes \ket{\psi_s}_{\RegB}
  \]
  for some $\alpha_s \ge 0$ and unit vectors $\ket{\psi_s}$.
  We say (a uniform family of) $\Gen$ is $p$-strongly hard if for every quantum polynomial-time adversary $\Adv$ (potentially with auxiliary input), its output has at most $p$ overlap with the correct state $\ket{\psi_s}$ given the classical part $s$, i.e.
  \[
      \E[\Tr(\proj{\psi_s}\Adv(s))] = \Tr\mparen{\sum_s \alpha_s^2 \proj{\psi_s} \Adv(s)} \le p
  \]
  for all sufficiently large input lengths.
\end{definition}

The strong hardness of the classical$\rightarrow$quantum extrapolation problem is equivalent to what is defined as $(1 - p)$-weak state puzzles in \cite{STOC:KhuTom25} where they further show the following.

\begin{theorem}[{\cite[Theorem 1.7]{STOC:KhuTom25}}]
  \label{import-thm:owpuzzle-extrapolation}
  For any $p$ inverse polynomially bounded away from $1$, one-way puzzles exist if and only if there exists a $p$-strongly hard classical$\rightarrow$quantum extrapolation problem.
\end{theorem}

\section{Our Main Theorem}

\newcommand{\As}{{\mathcal{A}}}

The setup is as follows. Alice and Bob are initially given a shared quantum state $\rho_0$. They then engage in an interactive protocol $\Pi$. We will divide the protocol into rounds, and in each round Alice sends some number $t$ of classical messages to Bob, and receives a corresponding $t$ answers from Bob. Because they may measure their state in order to generate their message, the state of Alice and Bob is continuously evolving. Let $\rho_i$ be the state after the $i$th message, so that the state after round $k$ is $\rho_{kt}$. Let $(\rho_i)_\As$ be the part of the state held by Alice. Let $T_i$ be the transcript from the first $i$ messages.

\paragraph{Impersonation attacks.} An impersonation attack on a protocol $\Pi$, which we will call Eve, is a quantum algorithm which views the classical messages being sent between Alice and Bob.
Eve does not know the initial state $\rho_0$. After some chosen round $k$ (chosen probabilistically by Eve), Eve produces a fake quantum state $(\rho'_k)_\As$ for Alice.  Then $(\rho_k)_\As$ is replaced with $(\rho_k')_\As$, and Alice and Bob interact for one more round. Let $T'_{kt+1}$ be the transcript of the communication, including the first $k$ rounds which were generated honestly, and then round $k+1$ which was generated using $(\rho_k')_\As$.

\paragraph{Our Theorem.} We now give our main theorem:

\begin{theorem}\label{thm:main}Let $n$ be the number of qubits in Alice's part of $(\rho_0)_\As$. Let $t$ be the number of messages exchanged in each round. Let $\epsilon$ be any desired inverse-polynomial probability. Then there is an impersonation attack Eve such that:
\begin{itemize}
    \item The number of passive rounds $k$ is at most $\lceil\frac{2nt}{\epsilon^2\ln2}\rceil$;
    \item The distributions $(k,T_{tk+1})$ and $(k,T'_{tk+1})$ are $\epsilon$-close in statistical distance.
\end{itemize}
Moreover, if Alice and Bob are efficient and one-way puzzles do not exist, then Eve is efficient.
\end{theorem}

Note that even though our theorem does not take into account either party's private quantum state, it can be easily extended to this case: consider a variant of the original protocol where at the end of those rounds, we have an additional round of communication where either Alice or Bob runs the optimal distinguisher and sends their result to the other party.

\begin{proof} To set up the proof, we first introduce the formal notation used to capture the honest authentication procedures. 

Let $\rho_0$ be the initial state over Alice and Bob. We can assume without loss of generality that $\rho_0$ is pure, by providing any purification to Bob; this does not affect the size of Alice's state.

We will partition the joint system $\rho_0$ into two registers $XY$ where Alice holds $X$ and Bob holds $Y$.
Furthermore, Alice will hold a register $Q_1$ that she will measure (in the computational basis) to obtain the first classical message to send to Bob.
In each iteration $i = 1,2,\cdots$, we start with $\rho_{i - 1}$ and apply the following operations to get $\rho_i$:
\begin{enumerate}
    \item Apply a coherent measurement isometry $CNOT_{Q_i \to Q_i, Q_i'}$.
    Let this state be named $\alpha_i$.
    \item Bob applies an isometry $M'$ with measurement on $Y Q_i$ (without loss of generality, we can assume $Q_i$ is first classically copied into $Y$), resulting in an updated register $Y$ and a classical random variable $A_i$. (Similarly, we capture the classical random variable by considering having a standard-basis copy of it in a separate register $A_i'$.)
    Let this state be named $\beta_i$. $A_i$ is given to Alice.
    \item Alice applies an isometry $M_{X A_i \to X, Q_{i + 1}}$ where similarly $A_i$ is also only used classically.
    (This step can be expanded into two steps like above, but this is unnecessary for the purpose of analyzing the $X$ part.)
\end{enumerate}
We will let $q_i$ denote the actual classical message sent to Bob, and $a_i$ the response. Thus, $T_i$ contains $(q_j,a_j)_{j\leq i}$.
Observe that for all $i$,
\begin{align*}
    H(XQ_i \mid Q_{\le i - 1}A_{\le i - 1})_{\rho_{i - 1}}
      &= H(X Q_i Q_i' \mid Q_{<i} A_{<i})_{\alpha_i} \\
      &\ge H(XQ_i \mid Q_{<i} A_{<i})_{\alpha_i} - H(Q_i' \mid Q_{<i} A_{<i})_{\alpha_i} \\
      &= H(XQ_i \mid Q_{<i} A_{<i})_{\alpha_i} - H(Q_i \mid Q_{<i} A_{<i})_{\alpha_i} \\
      &= H(X \mid Q_{\le i} A_{<i})_{\alpha_i} \\
      &= H(X \mid Q_{\le i} A_{<i})_{\beta_i} \\
    &\ge H(X \mid Q_{\le i} A_{\le i})_{\beta_i} \\
      &= H(X Q_{i + 1} \mid Q_{\le i} A_{\le i})_{\rho_i}.
    \end{align*}
    The first inequality is by Araki--Lieb (\Cref{lem:araki-lieb}), and the second inequality is due to the fact that conditioning on an additional classical variable $A_i$ cannot increase the conditional entropy.
    Also similarly, we can average the single-round decrement in the above chain to obtain an index $i$ with a small drop, just as in the simplified setting from the Technical Overview.
    Concretely, for each $i$ the difference
    \[
      H(XQ_iQ_i' \mid Q_{<i}A_{<i})_{\alpha_i} - H(X \mid Q_{\le i}A_{<i})_{\alpha_i}
      = H(XQ_i \mid Q_{\le i-1}A_{\le i-1})_{\rho_{i-1}} - H(X \mid Q_{\le i}A_{<i})_{\alpha_i}
    \]
    is upper bounded by the actual per-round drop
    $H(X \mid Q_{\le i-1}A_{\le i-1})_{\rho_{i-1}} - H(X \mid Q_{\le i}A_{\le i})_{\beta_i}$,
    and these drops telescope over $i=1,2,\ldots$ to at most $H(X)_{\rho_0}\leq n$.
    Therefore, we get that
    \begin{equation}
      \label{eq:auth-decrement-small}
      \sum_i H(XQ_iQ_i' \mid Q_{<i}A_{<i})_{\alpha_i} - H(X \mid Q_{\le i}A_{<i})_{\alpha_i} \le H(X)_{\rho_0}\leq n.
    \end{equation}
    As before, we now view $Y$ as a purifying “environment” for $X$ conditioned on the classical transcript revealed so far.
    Indeed, $\rho_0$ is pure on $X Y Q_1$, and every operation we used either copies a classical register by an isometry (the CNOT creating $Q_i'$) or measures into an explicitly recorded classical register (Bob’s $A_i$), so conditioning on $Q_{<i},A_{<i}$ leaves the joint state on $X Y Q_i Q_i'$ pure.
    Hence,
    $H(Y \mid Q_{<i}A_{<i})_{\alpha_i} = H(XQ_iQ_i' \mid Q_{<i}A_{<i})_{\alpha_i}$
    and
    $H(Y \mid Q_{\le i}A_{<i})_{\alpha_i} = H(X \mid Q_{\le i}A_{<i})_{\alpha_i}$,
    and \eqref{eq:auth-decrement-small} becomes
    \begin{align*}
        \sum_i I(Y;Q_i\mid Q_{<i}A_{<i})_{\alpha_i}
        &= \sum_i  H(Y \mid Q_{<i}A_{<i})_{\alpha_i} - H(Y \mid Q_{\le i}A_{<i})_{\alpha_i} \\
        &= \sum_i  H(XQ_iQ_i' \mid Q_{<i}A_{<i})_{\alpha_i} - H(X \mid Q_{\le i}A_{<i})_{\alpha_i} \\
        &\le H(X)_{\rho_0}\leq n.
    \end{align*}
    Intuitively, since the sum of these conditional mutual informations is bounded by the initial entropy $H(X)_{\rho_0}$, most increments must be small; in particular, there exists an index with $I(Y;Q_i\mid Q_{<i}A_{<i})$ small.
    By quantum Pinsker's inequality (applied to the state conditioned on the transcript), for such an index the fresh outcome $Q_i$ is almost independent of the environment $Y$ given $(Q_{<i},A_{<i})$, so $Q_i$ is simulatable from $Q_{<i}$.

Note that the role of $X, Q_i$ and $Y, A_i$ are symmetrical.
Therefore, a similar argument would also yield that
\begin{align*}
        \sum_i I(X;A_i\mid Q_{\le i}A_{<i})_{\alpha_i}
	&\le H(Y)_{\rho_0}=H(X)_{\rho_0}\leq n.
\end{align*}

Now, we consider the impersonation algorithm Eve as follows.
Eve chooses a random $k \gets [K]$ for $K \ge \frac{2nt}{\epsilon^2\ln2}$, while recording the classical query transcript.

After seeing the transcript $T_{kt}$ for the first $k$ rounds, the forging algorithm prepares the registers $XQ_{kt + 1}$ of state $\rho_{kt} \mid T_{kt}$.
This is the only step that might be inefficient.
Since this exactly corresponds to a classical$\to$quantum extrapolation problem, the whole attack is efficient if one-way puzzles do not exist.

\begin{lemma}\label{lemma:tracedistance}
  The trace distance between the real transcript and the transcript generated by the forging algorithm for the round chosen by the forging algorithm is at most $\sqrt{\frac{2t \cdot H(X)_{\rho_0}}{K\ln2}}\leq \sqrt{\frac{2tn}{K\ln2}}$.
\end{lemma}
\begin{proof}
  We first formalize the simulability guarantees via Pinsker and classical conditioning.
  For each step $i$, and for each fixed transcript prefix $(q_{<i},a_{<i})$, define
  \[
    \Delta_i^{Y}(q_{<i},a_{<i}) := \tfrac12\,\Big\|\big(\alpha_i\big)_{YQ_i\mid q_{<i},a_{<i}} - \big(\alpha_i\big)_{Y\mid q_{<i},a_{<i}}\otimes\big(\alpha_i\big)_{Q_i\mid q_{<i},a_{<i}}\Big\|_1,
  \]
  and similarly
  \[
    \Delta_i^{X}(q_{\le i},a_{<i}) := \tfrac12\,\Big\|\big(\beta_i\big)_{XA_i\mid q_{\le i},a_{<i}} - \big(\beta_i\big)_{X\mid q_{\le i},a_{<i}}\otimes\big(\beta_i\big)_{A_i\mid q_{\le i},a_{<i}}\Big\|_1.
  \]
  Quantum Pinsker (applied to each fixed transcript) gives
  \[
    \Delta_i^{Y}(q_{<i},a_{<i}) \le \sqrt{\tfrac1{2\ln2}\, I\big(Y;Q_i\mid q_{<i},a_{<i}\big)_{\alpha_i}},\qquad
    \Delta_i^{X}(q_{\le i},a_{<i}) \le \sqrt{\tfrac1{2\ln2}\, I\big(X;A_i\mid q_{\le i},a_{<i}\big)_{\beta_i}}.
  \]
  Averaging over the classical transcript and using Jensen's inequality (concavity of the square root), we obtain
  \[
    \mathbb{E}[\Delta_i^{Y}] \le \sqrt{\tfrac1{2\ln2}\, I(Y;Q_i\mid Q_{<i}A_{<i})_{\alpha_i}},\qquad
    \mathbb{E}[\Delta_i^{X}] \le \sqrt{\tfrac1{2\ln2}\, I(X;A_i\mid Q_{\le i}A_{<i})_{\beta_i}}.
  \]
  Summing over the $t$ steps in a run indexed by $k$ (i.e., $i=kt+1,\ldots,kt+t$) and applying Cauchy–Schwarz,
  \begin{align*}
    \mathbb{E}\left[\sum_{i=kt+1}^{kt+t} \Delta_i^{Y}\right]
      &\le \sum_{i=kt+1}^{kt+t} \sqrt{\tfrac1{2\ln2}\, I(Y;Q_i\mid Q_{<i}A_{<i})_{\alpha_i}}
      \le \sqrt{\frac t{2\ln2}}\sqrt{\sum_{i=kt+1}^{kt+t} I(Y;Q_i\mid Q_{<i}A_{<i})_{\alpha_i}},\\
    \mathbb{E}\left[\sum_{i=kt+1}^{kt+t} \Delta_i^{X}\right]
      &\le \sqrt{\frac t{2\ln2}\, \sum_{i=kt+1}^{kt+t} I(X;A_i\mid Q_{\le i}A_{<i})_{\beta_i}}.
  \end{align*}
  Averaging further over the random choice of $k\in[K]$ and using the bounds
  $\sum_i I(Y;Q_i\mid Q_{<i}A_{<i})_{\alpha_i} \le H(X)_{\rho_0}$ and
  $\sum_i I(X;A_i\mid Q_{\le i}A_{<i})_{\alpha_i} \le H(X)_{\rho_0}$ and Jensen's,
  we get
  \[
    \mathbb{E}_{k, q, a}\left[\sum_{i=kt+1}^{kt+t} \Delta_i^{Y}\right]
      \le \sqrt{\frac{t \cdot H(X)_{\rho_0}}{{2\ln2} \cdot K}},\qquad
    \mathbb{E}_{k, q, a}\left[\sum_{i=kt+1}^{kt+t} \Delta_i^{X}\right]
      \le \sqrt{\frac{t \cdot H(X)_{\rho_0}}{{2\ln2} \cdot K}}.
  \]
  
  \newcommand{\Hybrid}{\textsf{Hyb}}
  We consider the following sequence of hybrids $\Hybrid_0, \ldots, \Hybrid_t$, where $\Hybrid_i$ is defined as follows:
  \begin{enumerate}
    \item We start with running the original algorithm until step $kt + i$, obtaining $q_{kt + 1\ldots kt + i}, a_{kt + 1\ldots kt + i}$.
    \item We prepare a fresh copy of state $\rho_{kt + i} \mid q_{\le kt + i} a_{\le kt + i}$ conditioned on current transcript.
    We then replace the $XQ_{kt + i + 1}$ register with the corresponding register from the fresh copy.
    \item We finish running the algorithm, obtaining $q_{kt + i + 1\ldots kt + t}, a_{kt + i + 1\ldots kt + t}$.
  \end{enumerate}
  We can see that the real transcript corresponds to $\Hybrid_t$ and the transcript generated by the forging algorithm corresponds to $\Hybrid_0$.
  Therefore, it suffices to show that the trace distance between $\Hybrid_i$ and $\Hybrid_{i + 1}$ is small for all $i$.

  For each $i$, we further consider an additional hybrid $\Hybrid_i'$ where after the message $q_{kt + i + 1}$, we replace the $YA_{kt + i + 1}$ registers from a fresh copy of $\alpha_{kt + i} \mid q_{\le kt + i + 1} a_{\le kt + i}$ conditioned on current transcript.

  Observe that to bound the difference between $\Hybrid_i$ and $\Hybrid_i'$, it suffices to bound the trace distance between the two states immediately after $q_{kt + i + 1}$ is measured and, for $\Hybrid_i'$, after the $Y$ register is refreshed.
  By definition of $\Delta_i^{Y}$, the contribution of this replacement to the total trace distance is at most $\Delta_{kt+i+1}^{Y}$ by data processing inequality for trace distance, and similarly for the $XA$ replacement it is at most $\Delta_{kt+i+1}^{X}$.
  Summing over $i=0,\ldots,t-1$ and using the conditioned bounds on the sums of $\Delta$'s above, we conclude that
  \[
    \tfrac12\norm{\Hybrid_0 - \Hybrid_t}_1 \le \E_{k, q, a}\sum_{i=0}^{t-1} \big(\Delta_{kt+i+1}^{Y} + \Delta_{kt+i+1}^{X}\big)
      \le \sqrt{\frac{2t \cdot H(X)_{\rho_0}}{K\ln2}}.
  \]

\end{proof}
Plugging any $K\ge\frac{2nt}{\epsilon^2\ln2}$ into \Cref{lemma:tracedistance} gives \Cref{thm:main}.
\end{proof}

\section{Quantum Authentication Lower Bound}

\begin{definition}
  A quantum authentication scheme over classical channels consists of the following procedures:
  \begin{enumerate}
    \item \emph{Setup phase}: The two parties Alice and Bob initially share a joint quantum state $\ket{\Init}_{AB}$.
    \item \emph{Authentication phase}: Alice and Bob interactively exchange $t$ pairs of classical messages over a public authenticated classical channel (integrity without confidentiality) and at the end, Bob either accepts or rejects the authentication. This phase can be repeated in which case the two parties would run the same protocol using their leftover state from the previous round.
  \end{enumerate}

  We say that the scheme has $K$-time completeness $c$ if after $K$ rounds of authentication phases, Bob accepts all $K$ rounds with probability at least $c$.

  A $K$-time attack is a malicious party Eve who passively eavesdrops on all classical communication between Alice and Bob and interacts with Alice up to $K$ rounds of authentication but at some point, Eve needs to impersonate Alice to Bob in one of the $K$ rounds without seeing what Alice would have sent.
  The attack is considered successful if Bob accepts in that round.
  We say that the scheme has $K$-time soundness $s$ if for any $K$-time attack, the probability that Bob accepts is at most $s$.
\end{definition}

This soundness definition captures the information-theoretic setting.
In the computational setting, we instead consider a family of such schemes indexed by $\lambda$, and computational soundness is defined by restricting Eve to be QPT and requiring the soundness error to be negligible in $\lambda$.

\Cref{thm:main} almost immediately implies the following:

\begin{corollary}If there exists an unbounded-polynomial-round quantum authentication scheme over public authenticated classical channels that is computationally sound and has completeness negligibly close to $1$, then one-way puzzles exist.
\end{corollary}
\begin{proof}At the end of each authentication, we have Bob send a bit indicating whether he accepted or rejected. By completeness, each of these bits will be 1 with probability $c$. We then apply \Cref{thm:main}, which shows that Eve can choose some round and interact with Bob, causing Bob's acceptance bit to be 1 with probability close to $c$. In other words, Eve caused Bob to accept.
\end{proof}

\section{Application to Quantum Money}

\newcommand{\gen}{{\sf Gen}}
\newcommand{\ver}{{\sf Ver}}
\newcommand{\pk}{{\sf pk}}
\newcommand{\Os}{\mathcal{O}}
\newcommand{\money}{\mathord{\text{\$}}}

\subsection{Simulatable Oracles}

A classical oracle is an exponential-sized object, and as such, in general it cannot be constructed or evaluated efficiently. However, many oracles can be efficiently \emph{simulated}. We say that a distribution $\Os$ over oracles $O$ is efficiently simulatable if for any polynomial $q$ and any (potentially inefficient) algorithm $A$ making $q$ queries, there exists a quantum polynomial-time stateful simulator $S$ such that $|\Pr_{O\gets\Os}[A^O()=1]-\Pr[A^S()=1]|$ is negligible. 

Many oracles used for oracle separations in cryptography are simulatable. For example, random oracles and permutations are simulatable~\cite{C:Zhandry12,Zhandry25}. In fact, many oracles used in cryptography are built from random oracles/permutations plus other efficient computation, meaning such oracles are efficiently simulatable.

\begin{remark}Very recently and in completely unrelated contexts, several works~\cite{EC:LinMooWic25,TCC:DotMulSre25,TCC:GarGunWan25} show separations of certain primitives against all possible ``crypto'' oracles. These works similarly provide separations relative to extremely large families of functions. Note that their notion of crypto oracles are simulatable (since these oracles are efficient functions querying an internal random oracle), and both crypto and simulatable oracles are strict subsets of all possible oracles. Thus, our family of oracles is even broader than theirs.
\end{remark}

\subsection{Quantum Money Relative To Oracles, with Classical-Query Verifiers}

Here, we define public key quantum money relative to an oracle. For simplicity of notation, we implicitly have all procedures be functions of the security parameter, but do not explicitly denote the security parameter. The definition will largely follow that of~\cite{AC:AnaHuYue23}, but since it is different in important ways from the standard definition of quantum money, we first provide a discussion motivating the definition.

\paragraph{Oracles.} We will consider schemes defined relative to an oracle $O$ drawn from a distribution $\mathcal{O}$. All parties, including both the algorithms of the quantum money scheme and the adversary will be able to make queries to $O$.

\paragraph{Mini-schemes.} For simplicity, we will only consider a quantum money ``mini-scheme'', which roughly can be seen as a version of quantum money where the mint only ever creates a single banknote, and the adversary sees that banknote and tries to create two banknotes. Mini-schemes simplify the discussion, and in particular an impossibility for a mini-scheme implies an impossibility for a full scheme. Thus, for our purposes, considering mini-schemes only makes our results stronger. In the other direction, mini-schemes can be lifted to full schemes using digital signatures~\cite{STOC:AarChr12}, which can in turn be built from any one-way function. While in general the exact relationship between one-way functions and quantum money is unknown, we note that we expect most cryptographically-useful classical oracles to give one-way functions. Since we will be considering schemes which utilize such oracles, mini-schemes should therefore be considered essentially equivalent to full schemes in the setting we consider.

\paragraph{Classical query verifiers and correctness.} In the typical definition of quantum money, the scheme is considered correct as long as valid banknotes produced by the mint pass verification. Concretely, the definition only requires that the first verification of the banknote passes. But one may be worried that verifying a banknote actually destroys it, which would be a rather useless quantum money scheme. Fortunately, in the usual setting of quantum money, the standard one-time correctness notion actually implies that the banknote is preserved under verification, and in particular that subsequent verifications will pass. This is because correctness requires passing verification with overwhelming probability, and the Gentle Measurement Lemma~\cite{Winter99,Aaronson04} shows that such measurements can be performed in a way that negligibly affects the state.

In our setting, however, we cannot apply Gentle Measurement to get many-time correctness. This is because we are requiring the verifier to make only classical queries to the oracle, which requires measuring the quantum money state to get the query. Applying Gentle Measurements would turn this into a coherent process, meaning the verifier no longer makes only classical queries.

As such, following~\cite{AC:AnaHuYue23}, we simply stipulate in the correctness definition that the classical-query verifier will keep accepting the state for an arbitrary polynomial number of times. What this means is that, even though verification may alter the state, the perturbed banknote will still pass verification.

\begin{remark}Multi-time correctness truly is the ``right'' notion of correctness for quantum money, since we want to be able to verify banknotes many times. The only reason it is not standard in the literature is because for arbitrary verifiers it is possible to get away with the simpler one-time definition. 
\end{remark}

\begin{remark}Multi-time correctness is essential to our results as, for example, the verification of~\cite{STOC:AarChr12} can be modified to only make classical queries. Thus, there is an oracle relative to which there exists one-time-correct quantum money with a classical-query verifier. The problem is that the post-verification state under such a verifier will not pass a second verification. If one applies Gentle Measurements to this modified verifier, the resulting verifier is back to making quantum queries. 

Many-time correctness is also essential to the techniques of~\cite{AC:AnaHuYue23}. Here, however, note that it may be the case that even one-time correct quantum money is impossible from collision resistance\footnote{This is, after all, what~\cite{AC:AnaHuYue23} set out to prove.}. 
\end{remark}

We now give the definition.

\begin{definition}[Quantum Money, syntax] An oracle-aided public key quantum money mini-scheme relative is a pair of QPT oracle algorithms $\Pi^O=(\gen^O,\ver^O)$ where:
\begin{itemize}
    \item $\gen^O()$ samples a public key $\pk$ and a quantum state $\money$.
    \item $\ver^O(\pk,\money)$ takes as input $\pk,\money$, and outputs a bit $b$ together with a post-verification state $\money'$.
\end{itemize}
\end{definition}
We say that $\Pi^O$ has a classical query verifier (resp. mint) if $\ver$ (resp. $\gen$) only makes classical queries to $O$.

\begin{definition}[Quantum Money, Reusable Correctness] For a distribution $\Os$, an oracle-aided public key quantum money scheme $\Pi^O$ is \emph{reusably-correct} relative to $\Os$ if the following is true. For a polynomial $t$, consider the experiment $O\gets\Os$, $(\pk,\money_1)\gets\gen^O()$, and then $(b_i,\money_{i+1})\gets\ver^O(\pk,\money_{i})$ for $i=1,\cdots, t$. Then we require that, for all polynomials $t$, there exists a negligible $\epsilon$ such that $\Pr[b_1=b_2=\cdots=b_{t}=1]\geq 1-\epsilon$.
\end{definition}

\begin{definition}[Quantum Money, Oracle Security] For a distribution $\Os$, an oracle-aided public key quantum money scheme $\Pi^O$ is \emph{query-secure} relative to $\Os$ if, for all quantum algorithms $A$ making a polynomial number of classical queries to $O$, there exists a negligible function $\epsilon$ such that:
\[
  \Pr[
    \ver(\pk,\money_1)=\ver(\pk,\money_2)=1,\;
    O\gets\Os,\ (\pk,\money)\gets\gen^{\Os}(1^\lambda),\ \money_{1,2}\gets A^{\Os}(\pk,\money)
  ]\leq \epsilon(\lambda)\; ,
\]
where $\money_{1,2}$ is a joint system over two possibly entangled quantum money states $\money_1,\money_2$. 

$\Pi^O$ is \emph{computationally-secure} relative to $\Os$ if the above only holds for $A$ whose overall computation time is polynomial.
\end{definition}

\subsection{The impossibility}

\begin{theorem}For any distribution of oracles $\Os$, there is no public key quantum money scheme $\Pi^O$ such that:
\begin{itemize}
    \item $\Pi^O$ has a classical query verifier.
    \item $\Pi^O$ is reusably-correct
    \item $\Pi^O$ is query-secure.
\end{itemize}
Additionally, the last bullet can even be relaxed to computationally-secure, under the assumption that (1) one-way puzzles do \emph{not} exist and (2) $\Os$ is efficiently simulatable.
\end{theorem}
\begin{proof}We first describe a protocol between Alice and Bob. To initialize their joint states, sample $O\gets\Os$, and let $(\pk,\money)\gets\gen^O()$. Give $\money$ to Alice, and $\pk,O$ to Bob.

Now Alice and Bob interact as follows. Alice lets $\money_1=\money$. Then for $i=1,2\cdots$, Alice does the following:
\begin{itemize}
    \item Alice runs $(b_i,\money_{i+1})\gets\ver(\pk,\money_i)$, except that each time $\ver$ makes a (classical) oracle query $x_j$ to $O$, Alice sends $x_j$ to Bob. Bob, who knows $O$, computes $y_j=O(x_j)$ and sends $y_j$ to Alice.
    \item Alice sends the bit $b_i$ to Bob.
\end{itemize}
Notice that by the reusable correctness of the protocol, all of the bits $b_i$ that Alice sends to Bob will be 1 with overwhelming probability.

Let $\rho_i$ be Alice's state prior to the $i$th round of interaction. Notice that $\rho_i$ contains $\money_i$, but also any information produced during $\ver$ such as the queries and responses from $O$.

Now we use the impersonation attacker Eve. In general, Eve will be inefficient.
But if $\Os$ is simulatable, then Alice's view can be efficiently simulated.
If additionally one-way puzzles do not exist, then Eve will even be efficient: this is because the only potentially inefficient task Eve needs to solve is the extrapolation task from the classical transcript to Alice's internal state.

After some number $k$ iterations, Eve will be able to construct a quantum state $\rho_{k+1}'$. Part of $\rho'_{k+1}$ will be a simulated money state $\money_{k+1}'$, but $\rho'_{k+1}$ may contain other information as well. Now replace $\rho_{k+1}$ with $\rho'_{k+1}$. The guarantee from Eve is that, with probability at least (say) 1/2, the next round of interaction between Alice and Bob (but using Eve's $\rho'_{k+1}$) will result in $b_{k+1}=1$. In particular, since Alice interacting with Bob is simply running $\ver$, we have that $\Pr[\ver^O(\pk,\money_{k+1}')=1]\geq 1/2$.

From Alice and Eve, we can now describe our adversary $A$:
\begin{itemize}
    \item On input $(\pk,\money=\money_1)$, $A$ repeatedly runs $(b_i,\money_{i+1})\gets\ver^O(\pk,\money_i)$. Equivalently, $A$ runs Alice, but where messages to Bob are instead sent as queries to $O$.
    \item $A$ also runs Eve. At a step $k$ chosen by Eve, Eve will produce a quantum state $\rho'_{k+1}$, which in particular contains a money state $\money_{k+1}'$. 
    \item $A$ outputs $\money_{k+1},\money_{k+1}'$.
\end{itemize}
By correctness, $\Pr[\ver^O(\pk,\money_{k+1})=1]\geq 3/4$, and by the guarantees of Eve, $\Pr[\ver^O(\pk,\money'_{k+1})=1]\geq 1/2$. Thus, $\Pr[\ver^O(\pk,\money_{k+1})=\ver^O(\pk,\money'_{k+1})=1]\geq 1/4$, violating the security of $\Pi$.\end{proof}

\subsection{Interpretation}\label{sec:interpretation}

A major goal in cryptography is to argue that some cryptographic building block $P$ cannot be used to build another primitive $Q$. In general, we may believe both $P$ and $Q$ exist, so a trivial way to ``build'' $Q$ from $P$ is to have $Q$ simply ignore $P$, and use the assumed instantiation of $Q$. In order to argue that $Q$ cannot be built from $P$, we need to somehow restrict to instantiations of $Q$ that, in some sense, actually use $P$.

Following Impagliazzo and Rudich~\cite{STOC:ImpRud89}, the standard approach in cryptography is to use an oracle separation. Here, one gives an oracle $O$ relative to which $P$ provably exists, but $Q$ does not. Such oracle separations show that any construction technique that relativizes --- namely those that just query $P$ as a black-box --- and works relative to oracles cannot be used to build $Q$ from $P$. The vast majority of techniques are black-box, and so an oracle / black-box separation of this form rules out ``typical'' or ``natural'' techniques for building $Q$ from $P$. There are two typical formats of such separations, both sufficient for this purpose:
\begin{itemize}
    \item Give all algorithms free computation outside of the oracle, but define the notion of ``efficient'' algorithm to be one that only makes a polynomial number of queries to the oracle.
    \item Keep the notion of ``efficient'' as having polynomial overall computation time, but then provide an additional oracle which breaks some complexity class such as $\mathsf{NP}$ or $\mathsf{PSPACE}$. This breaks any realization of $Q$ that does not make queries to the oracle for $P$. As such, it forces $Q$ to ``use'' $P$.
\end{itemize}
Both formats are sufficient to show a black-box separations, but the latter is more fine-grained.

\paragraph{Comparison to~\cite{AC:AnaHuYue23}.} Our result significantly expands \cite{AC:AnaHuYue23} from random oracles --- which capture symmetric primitives such as one-way functions or collision resistance --- to all oracles, capturing essentially any cryptographic building block. When restricting to random oracles, our result also marginally improves on~\cite{AC:AnaHuYue23}. Namely,~\cite{AC:AnaHuYue23} follow the latter format, showing that there is no public key quantum money with classical verifier relative to a random oracle plus a $\mathsf{PSPACE}$ oracle. Since random oracles are efficiently simulate, our result shows that actually a random oracle plus an oracle breaking (oracle-free) one-way puzzles is sufficient. Such a one-way puzzle breaker can be achieved via $\mathsf{PP}$ and thus is potentially weaker than a $\mathsf{PSPACE}$ breaker.

\paragraph{Beyond black-box techniques.} Black-box techniques are those that do not relativize. These work by using the underlying circuit description of $P$ in building $Q$. Non-black-box techniques have been successfully used to overcome black-box impossibilities. A famous example is building identity-based encryption from cryptographic groups, which was proved impossible for black-box techniques in~\cite{EPRINT:PapRacVah12}, but was nevertheless shown possible with non-black-box techniques in~\cite{C:DotGar17}. 

A major limitation of most oracle separations is therefore that they cannot reason about non-black-box techniques, and therefore may falsely indicate a true impossibility. Despite this, oracle separations are nevertheless important for demonstrating large classes of techniques that fail. This may indicate an actual impossibility, or if not, it at least guides future work towards the techniques that can overcome the impossibility. 

It turns out, however, that many common non-black-box techniques can even be captured by appropriate oracles. This was first observed by Asharov and Segev~\cite{FOCS:AshSeg15} in the context of indistinguishability obfuscation (iO), and explored more generally by Zhandry~\cite{C:Zhandry22a}. Even though non-black-box techniques may make use of the underlying circuit representation, the technique is typically abstracted into a primitive that operates on circuits, such as zero knowledge proofs~\cite{STOC:GolMicRac85}, garbled circuits~\cite{FOCS:Yao86}, or indistinguishability obfuscation~\cite{C:BGIRSV01}. The central observation in these works is that the abstraction actually can be captured by an appropriate choice of oracle.

For example, consider the case of zero knowledge proofs, where one proves an $\mathsf{NP}$ statement without revealing anything about the underlying witness. A typical non-black-box technique in cryptography is to compute a zero knowledge proof about some statement involving a cryptographic primitive $P$. Zero knowledge proofs are non-black-box, in the sense that they work by transforming the $\mathsf{NP}$ statement involving $P$ into a circuit, and then operate on that circuit. However, this setting can be turned into an oracle as follows. The oracle has two components. The first component implements the primitive $P$, and the second component implements the zero knowledge proof. Importantly, to capture the non-black-box techniques, we need the zero-knowledge proof to operate on statements involving $P$. Real-world zero knowledge cannot accomplish this when $P$ is provided as an oracle. But in the relativized world, we can allow the zero knowledge proof to operate on oracle-aided statements that themselves make queries to $P$. Defining such oracles and proving that the do in fact realize the desired primitives is straightforward.

\paragraph{Interpreting our result.} Our oracle separation result holds relative to \emph{any} oracle. As such, it would hold relative to the $P$-plus-zero-knowledge oracles described above, or oracles implementing garbled circuits or indistinguishability obfuscation. In fact, any even non-black-box technique that can be modeled by an oracle taking as input oracle-aided circuits would also be handled by our impossibility. This captures the vast majority of techniques used in cryptography. This gives a much stronger separation than typical oracle separations, which are usually limited to a single oracle, and that oracle usually does not capture non-black-box techniques.

\section*{Acknowledgements}

During the preparation of this work, the authors used LLMs to generate initial draft text for certain sections.
The authors reviewed, revised, and take full responsibility for all content, ensuring its correctness and integrity.

\printbibliography

\end{document}